\documentclass{amsart}
\usepackage{tikz}

\newcommand{\alp}[1]{{\rm alph}(#1)}

\newcommand{\LL}{\mathcal L}
\newcommand{\comment}[1]{}

\newcommand{\mymod}{\hspace{3pt}{\rm mod}\,}
\newcommand{\MRR}[1]{}

\newcommand{\pref}{{\rm pref}}

\newtheorem{lemma}{Lemma}
\newtheorem{cor}{Corollary}

\newtheorem*{lemma*}{Lemma}

\usepackage{amsmath, amssymb, amsthm}

\newcommand{\NN}{\mathbb{N}}
\newcommand{\abs}[1]{\vert #1 \vert}

\newcommand{\ap} [2]{\ensuremath{\approx_{#1,#2}}}
\newcommand{\si} [2]{\ensuremath{\sim_{#1,#2}}}

\newcommand{\uv} [1]{``#1"}
\newcommand{\FW} {{\tt FW}}

\newenvironment{example}[1][{\rm \bf Example}]
{\begin{proof}[{\rm \bf #1}]
}
{\end{proof}
}

\begin{document}
\title{On an algorithm for multiperiodic words}
\author{\v St\v ep\' an Holub}
\address{Department of Algebra, Charles University, Sokolovsk\'a 83, 175 86 Praha, Czech Republic}
\email{holub@karlin.mff.cuni.cz}
\subjclass{68R15}
\keywords{periodicity, combinatorics on words}

\begin{abstract}
We consider an algorithm by Tijdeman and Zamboni constructing a word of length $n$ that has periods $p_1,\dots,$ $p_r$, and the richest possible alphabet. We show that this algorithm can be easily stated and its correctness briefly proved using the class equivalence approach.
\end{abstract}

\maketitle
\section*{A short (personal) history}
Non-trivial words with a given set of periods $P=\{p_1,p_2,\dots,p_r\}$ have received a lot of attention in the past decade. The motivation was to generalize the result by Fine and Wilf dealing with two periods, which became part of the folklore. A word with periods $P$ is called \emph{trivial} if $\gcd(P)$ is its period too. Papers \cite{CI} and \cite{TZ1} are two (independent) results considering non-trivial such words with maximal length and maximal cardinality of its alphabet. Those papers completed some older research of Castelli, Mignosi, Restivo and Justin (see, e.g., \cite{TZ2} for more details and references). Already in 1998, I wrote a short manuscript giving an analogous result (without considering its publication), which I showed to Sorin Constantinescu during the conference WORDS 2003 in Turku, where he presented their results. Since this was passed without notice in the subsequent publication and since I considered my approach simpler and more natural, I later decided to publish it in \cite{Holub1}. There was a gap in my paper, discovered by Gw{\'{e}}na{\"{e}}l Richomme, which is fixed in \cite{HolubCor}.

The present paper extends the same approach to the construction of the richest word with a given set of periods and a given length. The basic idea is to consider relations defined by the periods and understand  letters as (names of) equivalence classes generated by those relations. The idea is obvious and well known, usually expressed using the graph terminology (edges and connected components), rather than the algebraic terminology (relations and equivalence classes). Tijdeman and Zamboni \cite{TZ2} point out that the straightforward algorithm based on the graph approach is \uv{simple but inefficient} and then present an algorithm based on less transparent combinatorial analysis. The aim of this paper is to give a short description of their algorithm, as well as a short and intuitive proof of its correctness, using consistently the graph/equivalence viewpoint. 

\section{Notation}
Let $w$ be a word of length $n$ over an alphabet $A$. The set of all letters that occur in $w$ is denoted by $\alp w$. The $i$-th letter of $w$ is denoted by $w[i-1]$ so that $w=w[0]w[1]\cdots w[n-1]$. The prefix of $w$ of length $k$ is denoted by $\pref_k(w)$. 

We say that a positive integer $p$ is \emph{a period} of a word $w$ if $w[i]=w[i+p]$ for all $0\leq i \leq \abs w-i-1$ (where $|w|$ denotes the length of the word). Note that any $p\geq\abs w$ is a period of $u$. If $P$ is a set of positive integers such that each $p\in P$ is a period of $w$, we say that $w$ has periods $P$.  

The word of length $n$ having periods $P$ and the maximal possible cardinality of $\alp w$ is called an \emph{FW-word relative to $P$} (where FW stands for \uv{Fine and Wilf} for historic reasons). The word is called \emph{trivial with respect to $P$}  if $\gcd(P)$ is a period of $w$.
The longest non-trivial FW-word relative to $P$ is called an \emph{extremal FW-word relative to $P$}. We denote its length by $\LL(P)$ (note that $\LL(P)=L(P)-1$ where $L(P)$ is the notation adopted in \cite{TZ2}).


\section{Classes of Equivalence}
Let $w$ be a word which has periods $P$. For the rest of the paper we denote $m=\min P$.  Obviously, if $i\equiv j\mymod m$ or $\abs{i-j}\in P$, then $w[i]=w[j]$. These two conditions induce the relation $\si P k$ on integers $\{0,\dots,k-1\}$ defined by: 

$i\si Pk j$ if
\begin{itemize}
	\item 	$i\equiv j\mymod m$, or 
	\item there are integers $i',j'\in \{0,\dots,k-1\}$ such that $$i\equiv i'\mymod m, \quad j\equiv j'\mymod m$$  and $$\abs{i'-j'}\in P.$$
\end{itemize}
Let $\approx_{P,k}$ be the equivalence closure of $\si P k$. In other words, we have $i \approx_{P,k} j$ if and only if $i$ and $j$ lie in the same connected component of the graph defined by edges $i \si P k j$.  The class of $\ap P k$ containing $i$ will be denoted by $[i]_{P,k}$ and represented by its minimal element $\min [i]_{P,k}$.
 Then we obtain a word $\FW(P,k)$ of length $k$ over the alphabet $\NN$ by 
\begin{align*}
	\FW(P,k)[i]=\min [i]_{P,k}.
\end{align*}
The construction immediately yields that $\FW(P,k)$ is the unique (up to renaming of letters) FW-word of length $k$ relative to $P$.

\section{The algorithm}
The basic step of the algorithm is the reduction of $P$  to a new set of periods $Q$ defined by 
\begin{equation}\label{r}
Q=\{p-m \ \vert \  p\in P, p\neq m\}\cup \{m\}
\end{equation}
(where $m=\min P$ according to our convention). This reduction is, in fact, one step in the Eucledean algorithm, and is well known in the literature on multiperiodic words. The key fact about $P$ and $Q$ is expressed in the following lemma, which is an improved version of Lemma 2 from \cite{Holub1}.
\begin{lemma}\label{uk}
Let $k\geq 0$. Then for all $i,j\in\{0,1,\dots,k\}$
\begin{align*}
	[i]_{Q,k}=[j]_{Q,k} \quad \text{if and only if}\quad [i]_{P,k+m}=[j]_{P,k+m}.
\end{align*}
\end{lemma}
\begin{proof}
``$\Rightarrow$'': If $[i]_{Q,k}=[j]_{Q,k}$, then there is a sequence
 $i=i_0,\dots, i_\ell=j,$
of numbers from $\{0,1,\dots,k-1\}$ such that $$i_s\si Q{k} i_{s+1}$$ for each $s=0,\dots,\ell-1$. 
The relation $i_s\si Q{k} i_{s+1}$ implies $i_s\si P{k+m} i_{s+1}$, since either 
\begin{itemize}
	\item $i_s\equiv i_{s+1}\mymod m$, or 
	\item $\max\{i_s,i_{s+1}\}+m-\min\{i_s,i_{s+1}\}\in P$
\end{itemize}
Therefore $[i]_{P,k+m}=[j]_{P,k+m}$.

``$\Leftarrow$'': 
On the other hand, let $i=i_0,\dots, i_\ell=j,$ be a sequence of numbers from $\{0,\dots,k+m-1\}$, with $i,j\in \{0,\dots,k-1\}$, such that $$i_s\si P{k+m} i_{s+1}$$ for each $s=0,\dots,\ell-1$.
Certainly, we can suppose that the numbers in the sequence are pairwise distinct, whence $|i_s- i_{s+1}|\geq m$ and both $\min\{i_s,i_{s+1}\}$ and $\max\{i_s,i_{s+1}\}-m$ are in $\{0,1,\dots,k-1\}$. We now see that
$$\max\{i_s,i_{s+1}\}-m \si Q{k} \min\{i_s,i_{s+1}\}.$$
Therefore the sequence $$i=i_0,(i_0\mymod m),\dots,(i_\ell \mymod m),i_\ell=j$$ proves $[i]_{Q,k}=[j]_{Q,k}$.
\end{proof}
We have an immediate corollary.
\begin{cor}\label{cor}
For any $k\geq 0$, the word $\FW(Q,k)$ is a prefix of $\FW(P,k+m)$.
\end{cor}

The following lemma is an easy observation.
\begin{lemma}\label{obs}
Let $n-m\leq i\leq  m-1$. Then $[i]_{P,n}=\{i\}$.
\end{lemma}
\begin{proof}
Both $i-p$ and $i+p$ are out of range $\{0,1,\dots,n-1\}$ for any $p\in P$ (including $m$). Therefore $i$ is not related by $\si P{k}$ to any other element.
\end{proof}

From Corollary \ref{cor} and Lemma \ref{obs}, the formula 
\[\LL_P = m+\max\{\LL_Q,m - 1\}\]
can be readily derived (see \cite{Holub1, HolubCor}). In addition, it yields the following construction of $\FW(P,n)$, equivalent to Algorithm B described in \cite{TZ2}. 
\begin{enumerate}
	\item If $n\leq m$, then Lemma \ref{obs} with $k=0$ gives $$\FW(P,n)={\tt 0\cdot1\cdots (n-1)}.$$ (Recall that we consider integers as letters. To stress that, we use the typewriter font for them. The multiplication sign means concatenation). 
	\item Let $n>m$. Since the word $\FW(P,n)$ has a period $m$, it is determined by its prefix $w$ of length $m$.  Denote $u=\FW(Q,n-m)$. 
Corollary \ref{cor} and Lemma \ref{obs} imply that 
\begin{itemize}
	\item $w=\pref_m(u)$ if $m\leq n-m$, and 
	\item $w=u \cdot {\tt |u|\cdot (|u|+1)\cdots (m-1)}$ otherwise.
\end{itemize}
\end{enumerate}
 This can be succinctly stated as:
\[\FW(P,n)[i]=\left\{
\begin{array}{cl}
	\FW(Q,n-m)[i\mymod m] &\quad \text{if $(i\mymod m)<n- m$}\\
	i \mymod m&\quad  \text{otherwise}.
\end{array}
\right.\]

\begin{example}
Let $P=\{5,7\}$ and $n=8$.
Recursive definition of $\FW(P,8)$ leads to
\begin{align*}
	P=Q_0&=\{5,7\} & n=n_0&=8 \\
	 Q_1&=\{2,5\} & n_1&=3 \\
	 Q_2&=\{2,3\} & n_2&=1 
\end{align*}
In order to obtain the word $$u_0=\FW(Q_0,n_0)=\FW(P,8)$$ we will need words $$u_1=\FW(Q_1,n_1) \quad \text{ and}\quad u_2=\FW(Q_2,n_2).$$
Since $n_2=1$, we have $u_2=\text{\tt 0}$. From the point (2) above we have $$u_1=\pref_3(w_1^\omega)\quad\text{ where}\quad w_1={\tt 01}.$$ Therefore $u_1={\tt 010}$. Similarly, we get $$u_0=\pref_8(w_0^\omega) \quad\text{where}\quad w_0={\tt 01034},$$ whence $\FW(P,8)={\tt 01034010}$.

Schematically:
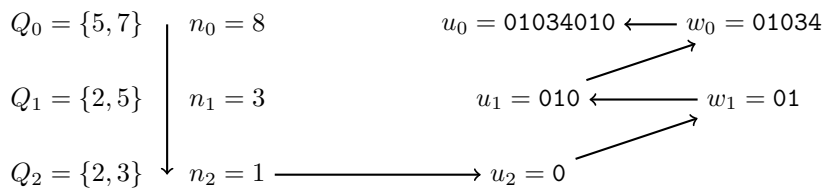
\[
\begin{tikzpicture}
\node (Q0) at (0,0) {$Q_0=\{5,7\}$}; 
\node (Q1) at (0,-1) {$Q_1=\{2,5\}$}; 
\node (Q2) at (0,-2) {$Q_2=\{2,3\}$}; 
\node (n0) at (2,0) {$n_0=8$}; 
\node (n1) at (2,-1) {$n_1=3$}; 
\node (n2) at (2,-2) {$n_2=1$}; 
\node (u0) at (6,0) {$u_0={\tt 01034010}$}; 
\node (u1) at (6,-1) {$u_1={\tt 010}$}; 
\node (u2) at (6,-2) {$u_2={\tt 0}$}; 
\node (w0) at (9,0) {$w_0={\tt 01034}$}; 
\node (w1) at (9,-1) {$w_1={\tt 01}$}; 
\draw[thick, ->] (1.2,0)--(1.2,-2);
\draw[->,thick] (n2)--(u2);
\draw[->,thick] (u2)--(w1);
\draw[->,thick] (w1)--(u1);
\draw[->,thick] (u1)--(w0);
\draw[->,thick] (w0)--(u0);
\end{tikzpicture}
\]
\end{example}

 From the above example we see that the procedure has two parts: \uv{descending} and \uv{ascending}, which are called \uv{Reduction} and \uv{Extension} in \cite{TZ2}. The end of reduction can be defined in several ways. We have seen that we can turn to extension as soon as we know $\FW(Q_i,n_i)$. This typically happens if $n_i\leq \min Q_i$, or if $\min Q_i=\gcd(Q_i)$.
\section{Concluding remarks}
As already remarked, the above algorithm is identical with Algorithm B from \cite{TZ2}. Even all arguments we use can be in some way traced back to similar arguments in literature. Nevertheless, I believe that the description presented here gives another evidence to the fact that the equivalence class approach is not only simple but also efficient and intuitive. (Another elegant example, in my opinion, is the proof of the fact that the extremal FW-word is a palindrome, given in \cite{Holub1}.)  

One possible drawback can be a bit discouraging notation like $\sim_{P,k}$, and the fact that notions like \uv{equivalence closure} may sound \uv{too algebraic} to some ears. Computer theorists could therefore like to translate the exposition into graph language and speak about edges instead of generating relations and about connected components instead of equivalence classes. The rest will be the same.

\bibliographystyle{plain}
\bibliography{multalg}

\end{document}